\documentclass{article}

\usepackage[main,final]{neurips_2026}

\usepackage[utf8]{inputenc} 
\usepackage[T1]{fontenc}    
\usepackage{hyperref}       
\usepackage{url}            
\usepackage{booktabs}       
\usepackage{amsfonts}       
\usepackage{nicefrac}       
\usepackage{microtype}      
\usepackage{xcolor}         

\usepackage{macro}
\usepackage{tikz}
\usetikzlibrary{positioning, arrows.meta, fit, backgrounds}
\usepackage{graphicx}

\newcommand{\paperappendixref}[2]{%
  \ifdefined\cameraready
    Appendix~#2 of the extended version of the paper~\citep{wu2026viverra}%
  \else
    Appendix~\ref{#1}%
  \fi
}
\newcommand{\paperappendicesref}[3]{%
  \ifdefined\cameraready
    Appendices~#3 of the extended version of the paper~\citep{wu2026viverra}%
  \else
    Appendices~\ref{#1} and~\ref{#2}%
  \fi
}
\title{\sys: Text-to-Code with Guarantees}

\author{
Haoze Wu$^{1,2}$, Rocky Klopfenstein$^{1}$,
Keith Farkas$^{2}$, Nina Narodytska$^{2}$
\\
{\small
$^{1}$Amherst College \quad
$^{2}$Broadcom
}
\\
{\small
\texttt{\{hwu,rklopfenstein27\}@amherst.edu}
}
\\
{\small
\texttt{\{keith.farkas,nina.narodytska\}@broadcom.com}
}
}

\begin{document}

\maketitle

\begin{abstract}
A fundamental limitation of \texttocode is that no guarantee can be obtained about the correctness of the generated code. Therefore, to ensure its correctness, the generated code still has to be reviewed, tested, and maintained by developers. However, parsing through LLM-generated code can be tedious and time-consuming, potentially negating the productivity gains promised by AI-coding tools. To address this challenge, we present \sys, a workflow that automatically produces \emph{formally verified annotations} alongside generated code to aid user's understanding of the generated program. Given a natural-language task description, \sys prompts an LLM to synthesize a C program together with candidate assertions expressing safety and correctness properties. It then verifies those assertions in a compositional and best-effort manner via a portfolio of bounded model checkers. Evaluation on 17 diverse programming tasks suggests that \sys can efficiently generate code with verified assertions, and that these assertions improve users' performance on code-comprehension tasks in a user study with 400 task completions.
\end{abstract}

\section{Introduction}\label{sec:intro}

\texttocode based on Large Language Models (LLMs) can significantly improve software development productivity. However, LLM-generated code is opaque: subtle bugs, incorrect logic, and violated assumptions can all slip through unnoticed. In practice, developers must manually review or rewrite LLM-produced programs before they can trust them, an overhead that risks negating the productivity gains promised by AI coding tools.
Since \texttocode does not by itself guarantee that the generated code correctly implements the natural language specification, human involvement is still unavoidable. We therefore argue that a key bottleneck in \texttocode is that developers lack efficient ways to check the generated code. In this paper, we ask the following question: \emph{How can we help the developer to understand and analyze AI-generated code?}

A standard answer to this question is to perform thorough \emph{testing}. However, generated code can be large, the logic can be non-trivial, and tests only witness correctness on a finite number of inputs; they cannot reliably characterize general program behavior across a wide range of execution traces. On the other hand, full formal verification, e.g., through proof assistants such as Lean or Rocq, can provide strong guarantees about specified program behavior. However, requiring a complete formal proof for every generated program is impractical: even with the aid of LLMs, program verification in proof assistants can still be time-consuming and challenging~\citep{first2023baldur,jiang2022draft}. Therefore, there is room for practical methods between these two extremes, that is, approaches that provide stronger guarantees than testing without requiring full formal verification.
To explore this gap, we propose \sys, a novel \texttocode framework. The goal of \sys is to provide developers with formally verified statements about bounded program behavior while keeping the cost of obtaining those facts practical.

The design of \sys makes two important compromises to achieve this goal. First, it puts forward the idea of \textbf{partial co-specification}: co-generating code together with properties of the code in the form of assertions. These properties capture requirements from the natural language description of the program functionality. Usually, they express safety and correctness properties of the code. These assertions serve simultaneously as partial specifications and as targets for formal verification.

Second, we design a \textbf{best-effort verification} procedure. We propose using bounded program verification to check these assertions in a best-effort manner: we do not require every assertion to be fully verified and we verify a reduced version of the program in which each loop runs for at most $K$ iterations. Our working hypothesis in this paper is that even incomplete verification is valuable: an assertion confirmed to hold on all executions within a bounded depth is a formally grounded fact that can help the developer understand the code and narrow what must be reviewed manually.

To test the feasibility of \sys, we perform an evaluation on 17 practical Text-to-C-Code tasks that goes beyond simple benchmarks seen in canonical datasets such as HumanEval~\citep{chen2021evaluating} and MBPP~\citep{austin2021program}. We find that \sys can efficiently generate complex programs with multiple verified assertions, and a  user study further suggests those verified facts are meaningful enough to boost program comprehension in practice.

To summarize, our contributions include the following:

\begin{itemize}
    \item We introduce \sys, an end-to-end \texttocode workflow that synthesizes C programs from natural-language descriptions and automatically produces formally verified annotations.

    \item We propose the idea of \emph{partial co-specification}: the LLM co-generates both a program and partial correctness properties from the same natural-language task description, turning those properties into assertions that serve as partial specifications.

    \item We design a \emph{best-effort} compositional verification procedure that uses call-graph structure and previously verified assertions as assumptions, allowing bounded model checking to produce both unconditional and conditional guarantees.

    \item We prove that every verified property is sound for the bound-reduced program within the stated loop-unwind bound, and that conditional facts are sound under their recorded dependency assumptions.

    \item We evaluate \sys on 17 diverse \texttocode benchmarks and show that it can generate and verify meaningful annotations efficiently; a user study further shows that these verified properties improve code comprehension compared to unannotated generated programs.
\end{itemize}

\section{Preliminaries}\label{sec:prelim}

We consider the following setting.
A developer provides a natural-language task description $D$.
An LLM produces a 
program $\prog$ intended to implement $D$.
Our goal is to automatically produce a set of \emph{verified facts}, formally grounded
statements about the behavior of $\prog$, that help the developer understand and
audit the generated code, without requiring the developer to write any specification.
In this section, we introduce the key concepts of \sys.

\begin{definition}[Program/Call graph]\label{def:program}
A \emph{program} $\prog$ consists of global code (headers, macros, type definitions) and
a finite set of functions $\mathcal{F}(\prog)$.
The \emph{call graph} $\callgraph{\prog}$ is a directed graph over $\mathcal{F}(\prog)$
where an edge $(f \to g)$ exists if $f$ calls $g$, annotated with the logical line
of the call site in $f$.
\end{definition}

\begin{definition}[Assertion]\label{def:assertion}
An \emph{assertion} is a triple $a = \tuple{f,\, l,\, \phi}$ where
$f \in \mathcal{F}(\prog)$ is a function name, $l \in \mathbb{N}$ is a \emph{line number} in $f$, a position between consecutive non-assertion statements, and $\phi$ is a Boolean predicate over program variables in scope at $l$.
We write $\assertions{\prog}$ for the set of all assertions in $\prog$.
The line numbers are counted over the assertion-free skeleton of $f$, so $l$ is stable under insertion or removal of assertions.
\end{definition}

We adopt the loop-unwinding semantics standard in bounded model checking~\citep{clarke2004cbmc}.
\begin{definition}[$k$-Bounded Execution]\label{def:kexec}
A \emph{$k$-bounded execution} of a program $\prog$ is any execution in which
every loop back-edge is traversed at most $k$ times; the model checker
asserts \texttt{false} at the $(k{+}1)$-th iteration to detect
violations of the bound.
We write $\execk{\prog}{k}$ for the set of all $k$-bounded executions of $\prog$.
\end{definition}

\begin{definition}[Bounded Invariant]\label{def:binv}
Given a program $\prog$, an assertion $a = \tuple{f, l, \phi}$, and a bound
$k \in \mathbb{N}$, we say $a$ is a \emph{$k$-bounded invariant} of $\prog$,
written $\binv{\prog}{k} \ni a$, if $\phi$ evaluates to \texttt{true} at line $l$
of $f$ on every execution in $\execk{\prog}{k}$.
\end{definition}

Bounded invariants are decidable and can be checked with a Bounded Model Checker via
Definition~\ref{def:kexec}; full invariance over all executions is undecidable
in general. A $k$-bounded invariant provides a formal guarantee for all
executions within the stated bound.

\begin{definition}[Assumption]\label{def:asm}
Given a program $\prog$ and a set of assertions
$S = \{\tuple{f_1,l_1,\phi_1}, \ldots, \tuple{f_n,l_n,\phi_n}\}$,
$\asm{\prog}{S}$ denotes the program obtained by treating each assertion in $S$
as an assumption,
which renders any execution on which $\phi_i$ is \texttt{false}
\emph{infeasible}, pruning it from the model checker's search space.
Formally:
\[
  \execk{\asm{\prog}{S}}{k}
  \;=\;
  \bigl\{\,e \in \execk{\prog}{k} \;\bigm|\;
    \forall\, \tuple{f_i,l_i,\phi_i} \in S,\;
    \phi_i \text{ holds at line } l_i \text{ of } f_i \text{ in } e
  \,\bigr\}.
\]
\end{definition}

An assumption restricts the set of executions considered by the model
checker to those consistent with $S$: only executions in which all $\phi_i$
hold are explored, without triggering assertion failures for those locations.

\begin{definition}[Verification Status]\label{def:status}
Given a program $\prog$, an ordered assertion sequence
$\mathit{t} = [a_1, \ldots, a_m]$, and a bound $k$, the
\emph{verification status} of $a_i$ is one of:
\begin{itemize}
  \item $\verified$: $a_i$ is a $k$-bounded invariant of $\prog$.
  \item $\cverified(S)$: $a_i$ is a $k$-bounded invariant of $\asm{\prog}{S}$,
        where $S$ is a set of assertions.
  \item $\falsified$: the model checker finds a $k$-bounded counterexample to $a_i$
        even under $S$.
  \item $\unknown$: verification times out.
\end{itemize}
\end{definition}

\section{Methodology}\label{sec:methodology}


\begin{figure}[t!]
\centering
\resizebox{\linewidth}{!}{
\begin{tikzpicture}[
    font=\scriptsize,
    node distance=0.55cm and 0.55cm,
    box/.style={
        rectangle,
        rounded corners,
        draw,
        thick,
        align=center,
        minimum height=0.78cm,
        minimum width=1.95cm,
        inner sep=2pt,
        fill=gray!8
    },
    llm/.style={
        box,
        fill=blue!8,
        inner sep=0pt
    },
    verify/.style={
        box,
        fill=orange!12
    },
    output/.style={
        box,
        fill=green!10
    },
    data/.style={
        rectangle,
        draw,
        thick,
        align=center,
        minimum height=0.72cm,
        minimum width=1.5cm,
        inner sep=1pt,
        fill=gray!10
    },
    group/.style={
        rectangle,
        rounded corners,
        draw,
        thick,
        dashed,
        inner sep=0.28cm,
        fill opacity=0.18,
        text opacity=1
    },
    arrow/.style={-{Stealth[length=2mm,width=1.4mm]}, thick},
    dashedarrow/.style={-{Stealth[length=2mm,width=1.4mm]}, thick, dashed}
]

\node[data] (desc) {NL task\\description $D$};

\node[llm, right=0.75cm of desc, yshift=1.3cm] (props) {
Property\\elicitation 
$\Pi$
};

\node[llm, right=1.15cm of desc] (annotate) {
Annotated\\program 
$P^+$
};

\node[llm, right=0.75cm of desc, yshift=-1.3cm] (synth) {
Program\\synthesis
$P_0$
};

\node[llm, right=0.72cm of annotate] (bound) {
Bound\\reduction 
$P^+_k$
};

\node[verify, right=0.65cm of bound, yshift=0.85cm] (cg) {
Call-graph\\traversal
};

\node[verify, right=0.65cm of cg, yshift=-0.85cm] (verifyall) {
Compositional\\verification
};

\node[data, right=0.35cm of verifyall] (statuses) {
Assertion\\statuses
};

\node[llm, right=1.2cm of statuses] (translate) {
Fact\\translation
};

\node[output, right=0.25cm of translate] (final) {
Code with\\verified facts
};

\node[data, above=0.75cm of final] (dev) {
Developer\\review
};

\begin{pgfonlayer}{background}
\node[
    group,
    fill=blue!20,
    fit=(props) (synth) (annotate),
    label={[font=\scriptsize\bfseries]above:Partial co-specification}
] (partialbox) {};

\node[
    group,
    fill=orange!25,
    fit=(bound) (cg) (verifyall) (statuses),
    label={[font=\scriptsize\bfseries]above:Best-effort verification}
] (verifybox) {};
\end{pgfonlayer}

\draw[arrow] (desc) -- (props);
\draw[arrow] (desc) -- (synth);

\draw[arrow] (props) -- (annotate);
\draw[arrow] (synth) -- (annotate);

\draw[arrow] (annotate) -- (bound);

\draw[arrow] (bound) -- (cg);
\draw[arrow] (cg) -- (verifyall);
\draw[arrow] (bound) -- (verifyall);

\draw[arrow] (verifyall) -- (statuses);

\draw[arrow] (statuses) -- node[above, align=center, font=\tiny] {
verified only
} (translate);

\draw[arrow] (translate) -- (final);
\draw[arrow] (final) -- (dev);

\draw[dashedarrow] (statuses.south) .. controls +(0,-0.85) and +(-1.2,-1.25) ..
    node[below, align=center, font=\tiny] {
    bounds / conditions
    } (final.south west);

\end{tikzpicture}
}
\caption{
Overview of \sys. 
}
\label{fig:overview}
\vspace{-0.2cm}
\end{figure}
Figure~\ref{fig:overview} shows an overview of the \sys workflow. 
In a nutshell, \sys takes a natural-language description as input and generates a program annotated with a set of formally verified facts, each tagged with the conditions and bound under which it was established. 
These additional facts make the task of analyzing and validating generated code easier for the developer, as our user study confirms. 
Next, we describe the workflow in detail.

\subsection{The \sys Pipeline}

Algorithm~\ref{alg:pipeline} describes the \sys workflow, and Figure~\ref{fig:overview} visualizes its main steps. 
The figure also highlights how two central ideas underlying the workflow, \emph{partial co-specification} and \emph{best-effort verification}, are implemented in the framework.
\sys starts by generating properties $\Pi$ (line~\ref{line:elicit}) and synthesizing the initial program $\prog_0$ (line~\ref{line:syn}) from $D$. 
The goal of the \textit{property elicitation} step is to generate safety and correctness properties that the program should satisfy based on the user description. 
As a result, it produces a set of properties
$\Pi = \{\pi_1, \ldots, \pi_n\}$ in natural language.
Figure~\ref{fig:cospec-example} shows an example of such a set for an excerpt from the `Bubble Sort' task: $\pi_1$ and $\pi_2$ are two properties for this task formulated in natural language.
These serve as a conceptual checklist for assertion generation. 
Next, \sys generates a C program $\prog_0$. 
Here, we instruct the LLM to perform \textit{verification-friendly synthesis}, e.g., to avoid recursion, with verifiability as a design constraint (see \paperappendixref{app:prompt-synth}{B.2} for the prompt).
Given $\prog_0$ and $\Pi$, \sys performs \textit{assertion generation} (line~\ref{line:annotate}).
It annotates $\prog_0$ with assertions that are expressed as \texttt{assert}
statements in C code at appropriate locations throughout the program, producing $\prog^+$.
Assertions may appear at function boundaries (preconditions at entry, postconditions
before returns), at loop invariant points, or wherever a property from $\Pi$ is most
naturally expressed.
Helper functions returning \texttt{bool} may be introduced for complex conditions.
Figure~\ref{fig:cospec-example} shows $\prog^+$ and examples of \texttt{assert} statements $a_1$ and $a_2$ that correspond to properties $\pi_1$ and $\pi_2$, where $a_2$ asserts a helper function.

\begin{algorithm}[t]
\scriptsize
\caption{The \sys Pipeline}\label{alg:pipeline}
\begin{algorithmic}[1]
\Require Natural-language description $D$, loop-unwind bound $k$
\Ensure Program $\prog_0$ annotated with verified facts
\State $\Pi \gets \Call{PropertyElicit}{D}$\label{line:elicit} \Comment{LLM: NL safety/correctness properties}
\State $\prog_0 \gets \Call{Synthesize}{D}$ \label{line:syn}\Comment{LLM: verification-friendly C program}
\State $\prog^+ \gets \Call{Annotate}{\prog_0,\, \Pi}$ \label{line:annotate}\Comment{LLM: embed \texttt{assert()} statements}
\State $\prog^+_k \gets \Call{BoundReduce}{\prog^+,\, k}$\label{line:reduce} \Comment{LLM: tighten loop-bound constants so loops iterate $\leq k$ times}
\State $\mathit{t} \gets \Call{CGTraversal}
{\prog^+_k}$ \label{line:call}\Comment{order assertions by call graph}
\State $(\mathit{IG},\, F) \gets \Call{VerifyAll}{\prog^+_k,\, \mathit{t},\, k}$\label{line:loop}
\State $M \gets \Call{MapToProperties}{\prog^+_k,\, \Pi}$ \Comment{LLM: assertion $\to$ property}
\State \Return $\Call{EmbedFacts}{\prog_0,\, \mathit{IG},\, M}$ \Comment{translate to NL facts in code}
\end{algorithmic}
\end{algorithm}

\begin{figure}[t!]
\centering
\resizebox{\linewidth}{!}{
\begin{tikzpicture}[
    font=\scriptsize,
    node distance=0.75cm and 0.65cm,
    box/.style={
        rectangle,
        rounded corners,
        draw,
        thick,
        align=left,
        inner sep=5pt
    },
    taskbox/.style={
        box,
        fill=white,
        text width=3.0cm
    },
    propbox/.style={
        box,
        fill=blue!8,
        text width=4.4cm
    },
    codebox/.style={
        box,
        fill=green!8,
        text width=5.6cm
    },
    arrow/.style={-{Stealth[length=2mm,width=1.4mm]}, thick},
    dashedarrow/.style={->, thick, dashed}
]

\node[taskbox] (task) {
\textbf{Task description $D$}\\[2pt]
``Write a program that implements the Bubble Sort algorithm to sort an array of integers in ascending order.''
};

\node[propbox, right=0.45cm of task,  yshift=2.0cm] (props) {
\textbf{Property elicitation $\Pi$}\\[3pt]
$\textcolor{red}{\pi_1}$: The global swap counter equals the total number of adjacent swaps.

$\textcolor{red}{\pi_2}$: After all passes complete, the array is sorted in non-decreasing order:
$
\forall k.\ 0 \le k \land k+1 < n \Rightarrow arr[k] \le arr[k+1].
$
};

\node[propbox, right=0.45cm of task,  yshift=-1.1cm] (syn) {
\textbf{Program synthesis $P_0$}\\[3pt]
\begin{minipage}{5.85cm}
\ttfamily
void bubble\_sort(int arr[], int n,\\
\hspace*{1.8em}int *total\_swaps) \{\\
\hspace*{1.2em}int i, j, swapped;\\
\hspace*{1.2em}int local\_total\_swaps = 0;\\[2pt]
\hspace*{1.2em}for (i = 0; i < n - 1; i++) \{\\
\hspace*{4.8em}...\\
\hspace*{1.2em}\}\\
\hspace*{1.2em}*total\_swaps = local\_total\_swaps;\\
\}
\end{minipage}
};

\node[codebox, right=5.65cm of task,  yshift=0.2cm] (code) {
\textbf{Annotated program $P^+$}\\[0pt]
\begin{minipage}{6.5cm}
\ttfamily
void bubble\_sort(int arr[], int n,\\
\hspace*{1.8em}int *total\_swaps) \{\\
\hspace*{1.2em}int i, j, swapped;\\
\hspace*{1.2em}int local\_total\_swaps = 0;\\[2pt]
\hspace*{1.2em}for (i = 0; i < n - 1; i++) \{\\
\hspace*{2.4em}swapped = 0;\\
\hspace*{2.4em}for (j = 0; j < n - 1 - i; j++) \{\\
\hspace*{3.6em}if (arr[j] > arr[j + 1]) \{\\
\hspace*{4.8em}int tmp = arr[j];\\
\hspace*{4.8em}arr[j] = arr[j + 1];\\
\hspace*{4.8em}arr[j + 1] = tmp;\\
\hspace*{4.8em}local\_total\_swaps++;\\
\hspace*{4.8em}swapped = 1;\}\\
\hspace*{2.4em}\}\\
\hspace*{2.4em}if (!swapped) break;\}\\
\hspace*{1.2em}*total\_swaps = local\_total\_swaps;\\
\hspace*{-.5em} $\textcolor{red}{a_1}$ assert(*total\_swaps == local\_total\_swaps);\\
\hspace*{-.3em}$\textcolor{red}{a_2}$ assert(is\_sorted\_non\_dec(arr, n));\\
\}
\end{minipage}
};

\draw[arrow]
    (task.east) --
    node[above, align=center, font=\scriptsize] {}
    (props.west);

\draw[arrow]
    (task.east) --
    node[above, align=center, font=\scriptsize] {}
    ([yshift=-0.35cm]syn.west);

\draw[arrow]
    (props.east) --
    node[above, align=center, font=\scriptsize] {}
    ([yshift=0.35cm]code.west);

\draw[arrow]
    (syn.east) --
    node[above, align=center, font=\scriptsize] {}
    ([yshift=-0.15cm]code.west);


\end{tikzpicture}
}
\caption{
Example of partial co-specification for a function in the `Bubble Sort' program. 
From the same task description $D$, \sys co-generates candidate properties $\Pi$ and an initial program $P_0$. 
The properties are then inserted into the program as assertions, producing $P^+$. 
The red labels show how each elicited property is linked to the corresponding generated assertion.
}
\label{fig:cospec-example}
\vspace{-0.25cm}
\end{figure}
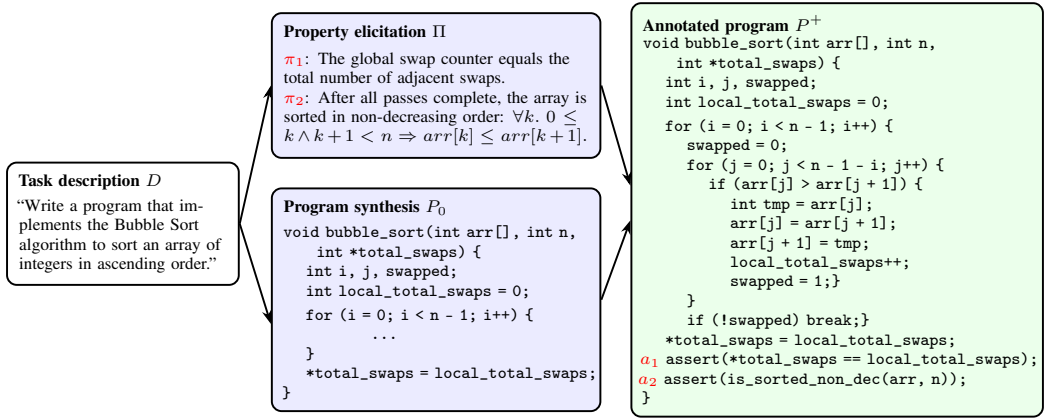

The next step is to perform best-effort verification of the generated assertions in $\prog^+$. 
We perform it in two phases. 
First, we perform a preprocessing step, \textit{bound reduction}, at line~\ref{line:reduce}. 
The goal is to make $\prog^+_k$ \emph{loop-bound tight}: no loop should reach its $(k{+}1)$-th iteration on any input to $\prog^+_k$.
As noted in Definition~\ref{def:kexec}, the model checker asserts \texttt{false}
at the $(k{+}1)$-th iteration of every loop; if any loop can legitimately run
beyond $k$ iterations, this unwinding assertion is violated, causing a spurious failure unrelated to the program's actual properties.
In practice, the LLM reduces compile-time loop-bound constants so that every
loop iterates at most $k$ times. 
Importantly, the verification results and Theorems~\ref{thm:soundness}--\ref{cor:chain} are stated over $\prog^+_k$ instead of the original program. 
Next, we perform \textit{compositional verification}, which we discuss in detail in the next section.

The final step is to compile facts for the developer, \textit{fact generation}, and present them to the user; this step is described in Section~\ref{sec:facts}.

\subsection{Compositional Verification}\label{sec:verification}

\sys has to verify programs with multiple annotations across different functions.
Therefore, we employ a compositional verification procedure that accounts for and leverages these interdependencies. 
To retrieve such dependencies, we first perform an analysis of the program and construct a call graph (line~\ref{line:call}). 
$\Call{CGTraversal}{}$ produces a linear ordering of $\assertions{\prog^+_k}$ by
performing a depth-first traversal of $\callgraph{\prog^+_k}$ from \texttt{main},
with cycle avoidance. 
For each function $f$, its assertions are collected in logical-line order.
When $f$ calls $g$ at logical line $l_c$, the traversal of $g$ is inserted
into $f$'s assertion list immediately before $f$'s assertions at lines $\geq l_c$.
By construction, for every call edge $(f \to g)$ at line $l_c$, all assertions
of $g$ precede all assertions of $f$ at lines $\geq l_c$ in $\mathit{t}$.
This ensures that callee postconditions are available as assumptions when
verifying caller assertions that follow the call.
Given $\mathit{t}$, we perform the main verification loop in line~\ref{line:loop}, outlined in Algorithm~\ref{alg:verify}.

Algorithm~\ref{alg:verify} processes assertions $a_i$ one by one in sequence order. 
The algorithm maintains an implication graph $\mathit{IG}$ that stores logical dependencies between the $i$th assertion and a preceding assertions in the order. 
We also keep track of falsified assertions in $F$. 
For each $a_i$, Algorithm~\ref{alg:verify}  attempts a \emph{standalone} check (line~\ref{line:alonecheck}): does $a_i$ hold with no additional assumptions?
We denote by $\verifier_k(\prog, \emptyset, a_i)$ a call to the automated reasoning tool that runs the bounded model checker on $\prog$ with $a_i$ as the target assertion and loop-unwind bound $k$. 
It returns one of $\{\mathsf{Verified}, \mathsf{Falsified}, \mathsf{Unknown}\}$~\footnote{In practice, $\unknown$ indicates that the solver exhausted its time budget without reaching a conclusion; such assertions are excluded from the output.}. 
If the check succeeds, we obtain a verified assertion with no dependencies. 
Otherwise, the algorithm attempts a \emph{compositional} check (line~\ref{line:composecheck}): does $a_i$ hold when all preceding assertions $S = \{a_j : j < i\}$ are promoted to \texttt{assume()} statements? 
Here, $\asm{\prog}{S}$ denotes the program obtained by treating each assertion in $S$ as an assumption. 
If this check succeeds, we obtain a verified assertion under the assumption that the assertions in $S$ hold. 
Otherwise, we mark it as failed and add it to $F$ (line~\ref{line:failedass}).
Promoting earlier assertions to assumptions is sound because they represent
facts about earlier program points that, if verified themselves, are guaranteed to hold. 
Next, we prove properties of \textsc{VerifyAll}. We denote $\mathtt{index}(\mathit{IG})$
the set of all keys in the dictionary structure $\mathit{IG}$.

\begin{restatable}[Soundness of \textsc{VerifyAll}]{theorem}{lemSoundness}\label{thm:soundness}
Let $\prog$ be a program, $\mathit{t} = [a_1, \ldots, a_m]$ an assertion sequence,
$k$ a bound, and $(\mathit{IG}, F) = \Call{VerifyAll}{\prog, \mathit{t}, k}$.
For every $i \in \mathtt{index}(\mathit{IG})$ with $a_i = \tuple{f_i, l_i, \phi_i}$:
\begin{enumerate}
    \item If $\mathit{IG}[i] = \emptyset$, then $\phi_i$ holds at line $l_i$ of $f_i$
          on every execution in $\execk{\prog}{k}$.
    \item If $\mathit{IG}[i] = S \neq \emptyset$, then $\phi_i$ holds at line $l_i$ of $f_i$
          on every execution in $\execk{\asm{\prog}{S}}{k}$.
\end{enumerate}
\end{restatable}
\begin{proof}
Both cases follow directly from the soundness of the underlying bounded model
checker (Definitions~\ref{def:kexec}--\ref{def:asm}):
$\verifier_k$ returns $\mathsf{Verified}$ only when it finds no counterexample
within bound $k$, meaning the assertion holds on all $k$-bounded executions of
the submitted program ($\prog$ in case~1; $\asm{\prog}{S}$ in case~2).
\end{proof}

The theorem establishes that each assertion is sound with respect to
assumptions $S$ that it depends on.

\subsection{Fact Generation}\label{sec:facts}

Since a dependency set of an assertion $a$ may itself contain independently verified and conditionally verified assertions,  we can tighten the set of assertions that $a$ depends on by iteratively back-tracing its dependencies to ``root'' assertions. We formalize this via the dependency closure.

\begin{definition}[Dependency Closure]\label{def:deps}
Given an implication graph $\mathit{IG}$ output by \Call{VerifyAll}{}, the
\emph{dependency closure} of assertion $a_i$ is defined inductively:
\[
  \deps{a_i} \;=\; \bigcup_{a_j \in \mathit{IG}[i]} \deps{a_j}.
\]
When $i \not\in \mathtt{index}(IG)$ (i.e., $a_i$ is unverified), $\deps{a_i} = \{a_i\}$; and when $\mathit{IG}[i] = \emptyset$ (i.e., $a_i$ is $\verified$), $\deps{a_i} = \emptyset$.
\end{definition}

Because Algorithm~\ref{alg:verify} records $\mathit{IG}$ explicitly, the
dependency closure $\deps{a_i}$ is checkable automatically.
We now prove that soundness is preserved under this tightened set of assumptions.

\begin{restatable}[Sound assumptions]{theorem}{thmChain}\label{cor:chain}
Under the hypotheses of Theorem~\ref{thm:soundness}, for every
$i \in \mathtt{index}(\mathit{IG})$ with $a_i = \tuple{f_i, l_i, \phi_i}$, 
  $\phi_i \text{ holds at line } l_i \text{ of } f_i
  \text{ on every execution in } \execk{\asm{\prog}{\deps{a_i}}}{k}$.
\end{restatable}
\begin{proof}
By structural induction on the definition of $\deps{a_i}$
(Definition~\ref{def:deps}).
The induction is well-founded because Algorithm~\ref{alg:verify}
ensures $\mathit{IG}[i] \subseteq \{a_j : j < i\}$, so the
dependency graph is acyclic.

\textbf{Base case}: $\mathit{IG}[i] = \emptyset$.
Then $\deps{a_i} = \emptyset$, and Theorem~\ref{thm:soundness}~(1)
gives that $\phi_i$ holds on every execution in
$\execk{\prog}{k} = \execk{\asm{\prog}{\emptyset}}{k}
= \execk{\asm{\prog}{\deps{a_i}}}{k}$.

\textbf{Inductive step}: $\mathit{IG}[i] = S \neq \emptyset$,
so $\deps{a_i} = \bigcup_{a_j \in S}\deps{a_j}$
by Definition~\ref{def:deps}.
By Theorem~\ref{thm:soundness}~(2), $\phi_i$ holds on every execution in
$\execk{\asm{\prog}{S}}{k}$.
It therefore suffices to show
\[
  \execk{\asm{\prog}{\deps{a_i}}}{k}
  \;\subseteq\;
  \execk{\asm{\prog}{S}}{k}.
\]
Let $e \in \execk{\asm{\prog}{\deps{a_i}}}{k}$ and $a_j \in S$.
We show $\phi_j$ holds on $e$:
\begin{itemize}
  \item If $j \notin \mathtt{index}(\mathit{IG})$, then
        $\deps{a_j} = \{a_j\}$, so $a_j \in \deps{a_i}$, and $\phi_j$
        is assumed in $\asm{\prog}{\deps{a_i}}$; hence $\phi_j$ holds on $e$.
  \item If $\mathit{IG}[j] = \emptyset$, then
        Theorem~\ref{thm:soundness}~(1) gives $\phi_j$ on all of
        $\execk{\prog}{k} \ni e$.
  \item If $\mathit{IG}[j] \neq \emptyset$, the inductive hypothesis
        gives $\phi_j$ on every execution in
        $\execk{\asm{\prog}{\deps{a_j}}}{k}$.
        Since $\deps{a_j} \subseteq \deps{a_i}$, we have
        $e \in \execk{\asm{\prog}{\deps{a_i}}}{k}
           \subseteq \execk{\asm{\prog}{\deps{a_j}}}{k}$,
        so $\phi_j$ holds on $e$.
\end{itemize}
Since $\phi_j$ holds on $e$ for every $a_j \in S$, the execution $e$
satisfies all assumptions of $\asm{\prog}{S}$, i.e.,
$e \in \execk{\asm{\prog}{S}}{k}$.
Therefore $\phi_i$ holds on every execution in
$\execk{\asm{\prog}{\deps{a_i}}}{k}$.
\end{proof}


After verification and dependency analysis, each $\verified$ or $\cverified$ assertion is translated by  $\sys$ into a natural-language statement and embedded as a comment in $\prog_0$, where $K$ is the bound under which it was established:
\begin{itemize}
  \item A $\verified$ statement is tagged \texttt{//FACT[k=\textit{K}]}.
  \item A $\cverified$ statement is tagged \texttt{//FACT[k=\textit{K}, cond]} and
        is accompanied by the index set $\deps{a_i}$
        (Definition~\ref{def:deps}). 
\end{itemize}
By Theorem~\ref{cor:chain}, every presented
fact is formally grounded: a developer who accepts the stated conditions for
$\cverified$ facts obtains a sound guarantee for all $k$-bounded executions
of the program. These fact-annotated programs are what developers receive and empirically 
improve developers' comprehension of the program.

\begin{algorithm}[t]
\scriptsize
\caption{\textsc{VerifyAll}: Compositional Verification}\label{alg:verify}
\begin{algorithmic}[1]
\Require Program $\prog$, assertion sequence $\mathit{t} = [a_1, \ldots, a_m]$, bound $k$
\Ensure Implication graph $\mathit{IG}$, falsified set $F$
\State $\mathit{IG} \gets \{\}$;\quad $F \gets \{\}$
\For{$i \gets 1$ \textbf{to} $m$}
    \State $r \gets \verifier_k\bigl(\prog,\ \emptyset,\ a_i\bigr)$ \label{line:alonecheck}
        \Comment{standalone: no assumptions}
    \If{$r = \mathsf{Verified}$}
        \State $\mathit{IG}[i] \gets \emptyset$
    \Else
        \State $S \gets \{a_j : j < i\}$
        \State $r \gets \verifier_k\bigl(\asm{\prog}{S},\ \emptyset,\ a_i\bigr)$\label{line:composecheck}
            \Comment{compositional: assume all preceding}
        \If{$r = \mathsf{Verified}$}
            \State $\mathit{IG}[i] \gets S$
        \ElsIf{$r = \mathsf{Falsified}$}
            \State $F \gets F \cup \{i\}$ \label{line:failedass}
        \EndIf
    \EndIf
\EndFor
\State \Return $(\mathit{IG},\, F)$
\end{algorithmic}
\end{algorithm}

\section{Implementation}\label{sec:impl}

We implemented a prototype of \sys in Python. We now describe the key implementation choices for \sys's LLM backend and verification engine.

\noindent\textbf{LLM backend.}
All prompting stages (property elicitation, synthesis, annotation, bound reduction,
assertion-to-property mapping, and fact translation) call the OpenAI API.
Responses are cached by a hash of the prompt and model so that repeated runs are
reproducible and cost-free.
Reasoning effort is set to \texttt{low} for synthesis and annotation to reduce latency,
and can be increased for harder benchmarks.

\noindent\textbf{Portfolio model checker.}
$\verifier_k$ is implemented as a parallel portfolio of two bounded model checkers:
CBMC~\citep{clarke2004cbmc} (run with and without the Bitwuzla~\citep{niemetz2023bitwuzla} SMT back-end) and
ESBMC~\citep{cordeiro2012esbmc} (with Bitwuzla).
All three solver configurations are launched concurrently;
the portfolio returns as soon as any solver reaches a definitive answer.
If all solvers time out within a per-assertion budget, the status is $\unknown$.
We use a default unwind bound of $k = 5$ and a per-solver timeout of 60\,s.

\noindent\textbf{Assertion normalization.}
Before handing $\prog^+$ to the model checker, \sys flattens multi-line
\texttt{assert(\ldots)} statements onto a single line.
This ensures that the logical line numbers in Definition~\ref{def:assertion}
correspond exactly to the source positions seen by the parser, and that
libclang can unambiguously identify each assertion's location.

\noindent\textbf{Call-graph extraction.}
\sys parses $\prog^+_k$ with libclang to obtain the call graph and per-function
source extents.
Assertions are extracted by regex, stripped from the function body to compute
logical line numbers (Definition~\ref{def:assertion}), and re-inserted during
verification as either \texttt{assert} or \texttt{assume} by the
\texttt{Method.generate\_with\_assertions} routine.
Global code is preserved verbatim and prepended to every verification query.

\noindent\textbf{Compilation checks.}
Between every LLM-generation step, \sys compiles the output with \texttt{gcc}
and re-prompts on failure, feeding the compiler error back as additional context.
This ensures that both $\prog_0$ and $\prog^+$ are syntactically valid before
any verification is attempted.

\section{Experimental Evaluation}\label{sec:eval-sat}

In this section, we perform experimental evaluation of \sys on a diverse set of \texttocode tasks. Concretely,
we evaluate \sys along two axes:
\begin{itemize}
    \item (\emph{Feasibility}): can \sys generate programs and derive verified annotations for
non-trivial \texttocode tasks?
\item (\emph{Utility}): do those verified annotations help developers understand
the generated code?
\end{itemize}

\subsection{Experimental Setup}

We evaluate on 17 diverse \texttocode tasks, from implementing the Simplex algorithm for deciding the feasibility of linear systems to parsing and filtering lines in a log file.  Each task is specified by a natural-language description, which can be found in \paperappendixref{app:benchmarks}{A};  \sys takes the description as input and automatically generates the corresponding program and verified assertions. We use GPT-5.1 with low reasoning for LLM calls and set the loop-unwind bound to $k = 5$. The portfolio verification engine uses a per-assertion wall-clock timeout of 60\,s. Experiments were run on a workstation equipped with AMD Ryzen Threadripper PRO 5945WX processors.


\subsection{Performance of \sys on program-annotation co-synthesis}

Table~\ref{tab:verification} shows the per-task breakdown of assertion statuses.
Depending on the complexity of the generated code, \sys generates 10 to 85 assertions.
The vast majority of generated assertions are verified by the model checker without any
assumptions, confirming that \sys can produce programs with assertions that are
automatically verifiable. A small fraction cannot be verified in isolation but hold
conditionally on earlier established facts. For most benchmarks the number of
unverified assertions (falsified or timed out) is small, but \textit{base\_addition}
is a clear outlier with 14 of 23 assertions falsified. Upon closer examination, these
failures are not due to bugs in the program logic, but due to the over-approximation of 
certain operators by the verification tool which led to spurious counter-examples. 

LLM synthesis time ranges from 36.6\,s to 459.3\,s across benchmarks, which is
reasonable given that each run involves multiple sequential LLM invocations
(property elicitation, synthesis, annotation, bound reduction).
Verification time varies considerably more, from 2.4\,s (\textit{shift\_letter}) to
2077.6\,s (\textit{sudoku}), driven by the complexity of the properties and programs.
For example, \textit{sudoku} timed out on assertions like
\texttt{board\_is\_partial\_valid}, which requires the model checker to reason
about global row, column, and subgrid uniqueness simultaneously; similarly,
\textit{knapsack}'s costliest assertion checks that the entire item array is
sorted by density after a comparison-based sort.
Such global relational properties are inherently harder for bounded model
checking than local invariants. Nonetheless, as shown in the next section, \sys was able to 
generate and verify meaningful assertions that help developers with program understanding.

The current implementation verifies assertions sequentially, so total
verification time grows linearly with the number of assertions.
Parallelizing the verification loop---running independent assertion checks
concurrently---is the most direct path to reducing wall-clock time, since
each call to $\verifier_k$ is independent.
More aggressive strategies include tightening the loop-unwind bound for
assertions identified as expensive, or limiting the assertion budget per
function to focus verification effort on the most informative properties.

\begin{table}[t]
\centering
\caption{Verification results per benchmark. \emph{LLM (s)} is total wall-clock time for
all LLM calls (synthesis, annotation, bound reduction, mapping, translation);
\emph{Verify (s)} is total verification time.}
\label{tab:verification}
\scriptsize
\begin{tabular}{lcccccc}
\toprule
\textbf{Benchmark} & \textbf{\#Assert} & \textbf{\verified} & \textbf{\cverified}
                   & \textbf{\textbf{Unverified}} & \textbf{LLM (s)} & \textbf{Verify (s)} \\
\midrule
bubblesort             &  10 &  10 &  0 &  0 &   36.6 &    2.9 \\
condorcet              &  25 &  22 &  0 &  3 &  172.5 &  253.0 \\
connected\_component   &  21 &  21 &  0 &  0 &   95.7 &  127.5 \\
shortest\_path         &  25 &  24 &  0 &  1 &  143.0 &   32.8 \\
knapsack               &  23 &  15 &  4 &  4 &  125.3 &  879.3 \\
sudoku                 &  49 &  43 &  0 &  6 &  164.9 & 2077.6 \\
tictactoe              &  85 &  85 &  0 &  0 &  182.0 &   21.6 \\
connect4               &  59 &  56 &  0 &  3 &  459.3 &  981.0 \\
numberle               &  34 &  32 &  0 &  2 &  345.3 &    8.6 \\
gaussian\_elimination  &  17 &  17 &  0 &  0 &   89.8 &  139.7 \\
simplex                &  21 &  21 &  0 &  0 &  234.0 &  163.2 \\
log\_timestamp\_filter &  20 &  18 &  1 &  1 &  223.0 &    9.3 \\
sort\_timestamp        &  34 &  34 &  0 &  0 &  121.3 &   24.0 \\
shift\_letter          &  19 &  19 &  0 &  0 &  104.9 &    2.4 \\
base\_addition         &  23 &   9 &  0 & 14 &   44.5 &   12.6 \\
cellular\_automata     &  17 &  17 &  0 &  0 &  159.3 &  206.9 \\
maze                   &  32 &  31 &  0 &  1 &   65.4 &  436.2 \\
\midrule
\textbf{Mean$\pm$Std}  & $30.2{\pm}18.6$ & $27.9{\pm}19.0$ & $0.3{\pm}1.0$ & $1.3{\pm}3.4$ & $162.8{\pm}108.0$ & $316.4{\pm}542.7$ \\
\bottomrule
\end{tabular}
\vspace{-0.5cm}
\end{table}

\subsection{Can verified annotations aid comprehension?}

We conducted a user study to evaluate whether verified annotations improve code comprehension. Concretely, we manually generated 10 multiple-choice questions, each anchored to a single verified assertion. Questions were designed to reference the underlying assertion without quoting it directly (e.g., \emph{``When the program prints the vertices of a connected component (line 159), what order are they in?''}). We then manually constructed three distractors per question, leading to a four-choice multiple-choice format. The raw surveys can be found in the supplementary materials.

Each multiple-choice question was evaluated using Human Intelligence Tasks (HITs) on Mturk. Each HIT consisted of a single survey containing basic instructions, a C program with lines numbered for readability, and one multiple choice question. For each question, we consider two cases, one where the program is not annotated with verified facts, the other where the C program was annotated with \texttt{//FACT:} comments, which are natural language statements translated from verified assertions. For this experiment, we only added assertions that are independently verified (without assumptions). These translations were intended to preserve only the logic described by an assertion. The survey was set up such that there is no overlap between participants of the two HITs with respect to the same question. Each HIT was completed by at least 20 workers. All HITs shared the same basic worker qualifications: a HIT approval rate $>95\%$, a minimum of 1{,}000 approved HITs, and the MTurk job function qualifier \textit{Information Technology}, targeting tech-adjacent participants. Each worker was allowed 3 minutes to complete the survey and paid \$0.50 per HIT.

Figure~\ref{fig:results} summarizes the results. The top-left plot shows the fraction of participants that answer correctly for each question (accuracy). The accuracy is higher when the programs are annotated with verified facts on all but two questions. This suggests that the verified annotations denote information that is not immediately obvious to the participants and that can improve code comprehension.

The remaining three plots examine response timing. Participants overall take less time to answer a question when annotations are present (top right). Among participants who answered correctly (bottom left), Group B is faster on 6/10 benchmarks, in particular, Group B participants are on average faster on  the last two questions where the accuracy between Group A and group B is the same. Interestingly, on the first question (connected-comps[Vertice order]), Group B users spent significantly longer (96s vs. 82s) while achieving a significantly higher accuracy (100\% vs. 40\%). We speculate this is due to the overhead the user spent on reading the annotations. Overall, our experiment suggests that \sys can improve practical code comprehension both in terms of accuracy in answering questions about the program and the response speed.

\begin{figure}[t!]
    \centering

    \begin{center}
    \small
    \tikz[baseline=0.1ex]{\fill[color={rgb,255:red,76;green,114;blue,176}, opacity=0.85] (0,0) rectangle (0.22,0.22);}
    \hspace{0.35em} Group A (unannotated)
    \hspace{1.6em}
    \tikz[baseline=0.1ex]{\fill[color={rgb,255:red,221;green,132;blue,82}, opacity=0.85] (0,0) rectangle (0.22,0.22);}
    \hspace{0.35em} Group B (annotated)
    \end{center}

    \vspace{0.35em}

    \begin{subfigure}[b]{0.49\linewidth}
        \includegraphics[width=\linewidth, trim=0 110 0 0, clip]{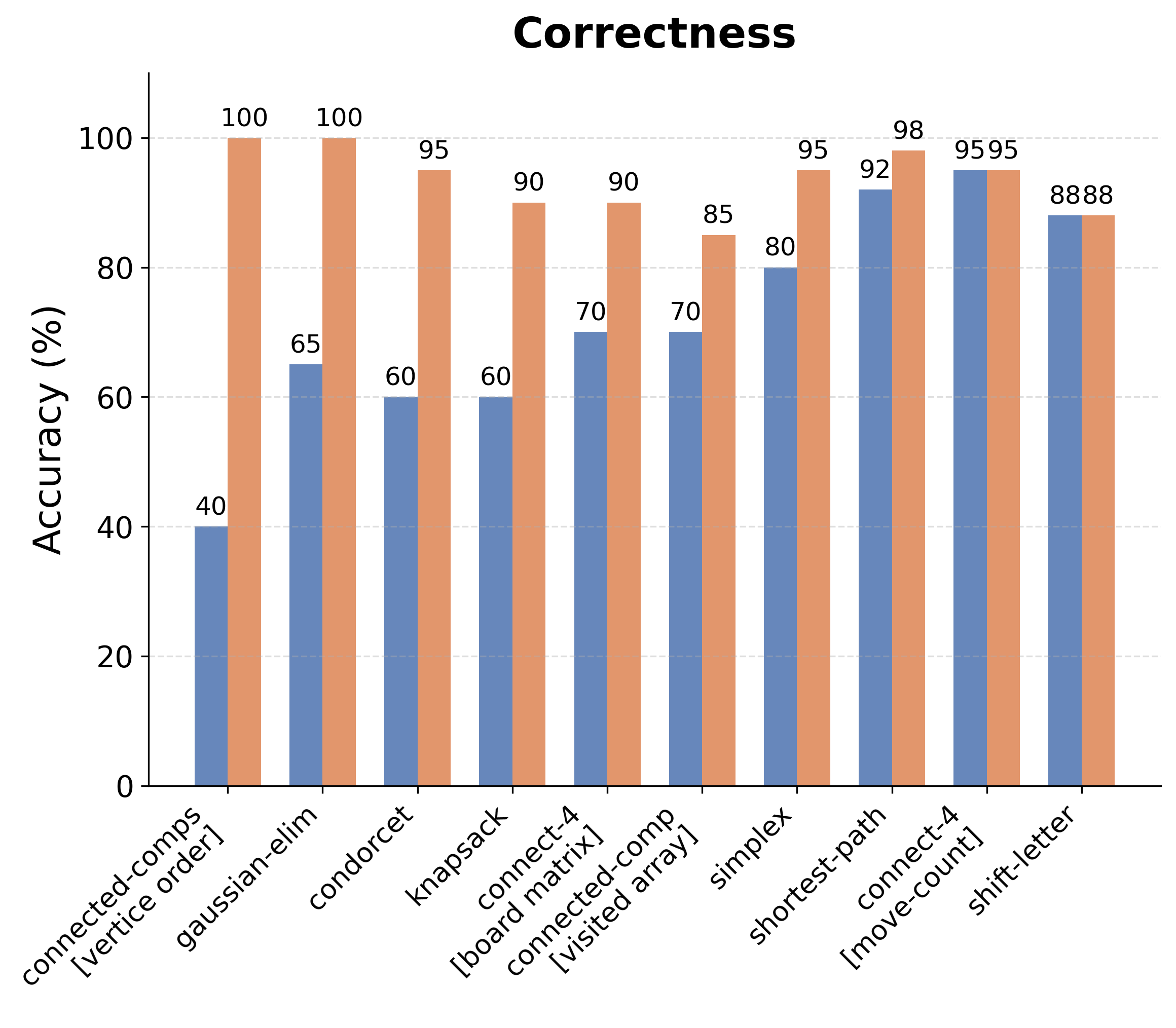}
    \end{subfigure}
    \hfill
    \begin{subfigure}[b]{0.49\linewidth}
        \includegraphics[width=\linewidth, trim = 0 110 0 0, clip]{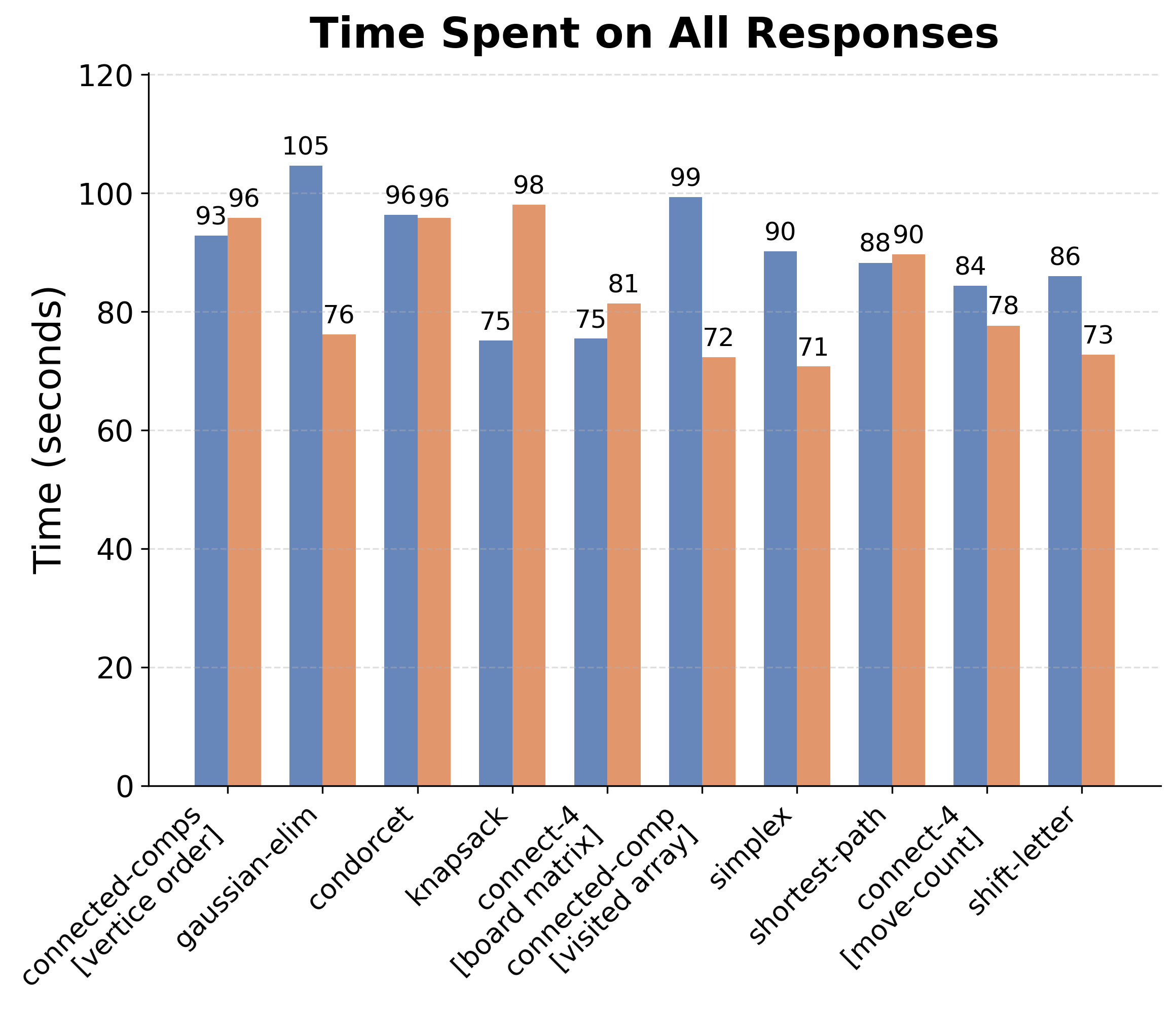}
    \end{subfigure}

    \vspace{0.75em}

    \begin{subfigure}[b]{0.49\linewidth}
        \includegraphics[width=\linewidth]{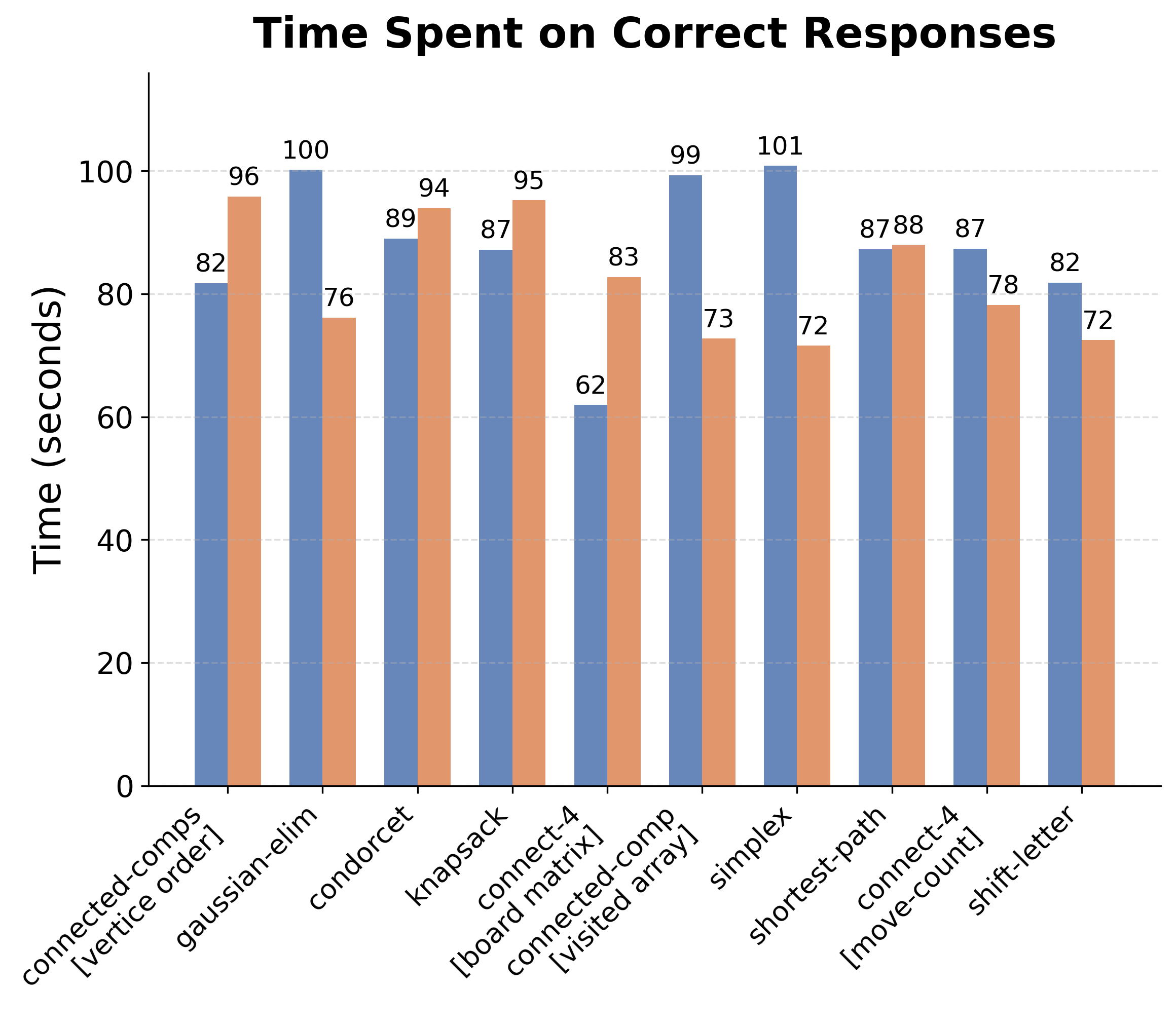}
    \end{subfigure}
    \hfill
    \begin{subfigure}[b]{0.49\linewidth}
        \includegraphics[width=\linewidth]{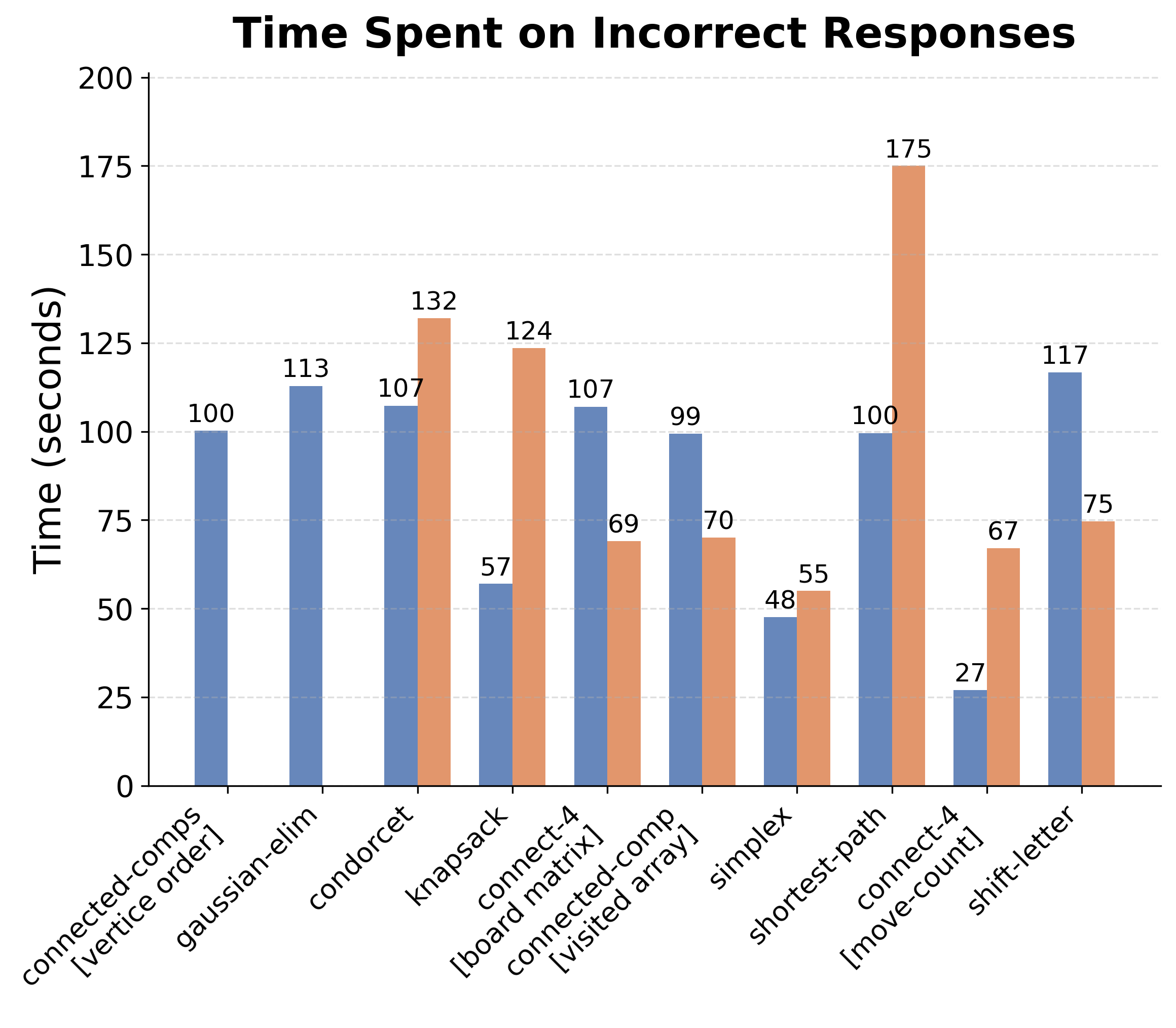}
    \end{subfigure}

\vspace{-0.3cm}
    \caption{Comparison of correctness and user-timing between treatment and control HITs. The plots visualize overall correctness (top-left) and total completion time (top-right), with further breakdowns of timing for correct (bottom-left) and incorrect (bottom-right) responses across all benchmarks.}
    \label{fig:results}
    \vspace{-0.5cm}
\end{figure}

\section{Related Work}\label{sec:related}

\paragraph{LLM-based code synthesis.}
LLMs have demonstrated strong code generation capability across benchmarks spanning
algorithmic problems~\citep{chen2021evaluating,austin2021program,li2022competition}
and real-world repositories~\citep{jimenez2024swebench}.
These systems provide no correctness guarantees; developers must audit the output themselves.
\sys is orthogonal to any synthesis back-end: it augments generated code with
formally verified annotations rather than improving generation quality.

\paragraph{Verified code generation.}
A growing line of work pursues code generation with formal correctness guarantees.
\citet{misu2024towards} evaluate LLMs on synthesizing Dafny methods with full contracts;
\citet{chen2024automated} synthesize Verus proof annotations for Rust via self-evolving fine-tuning;
Clover~\citep{sun2024clover} iteratively verifies LLM-generated Dafny in a closed loop;
Baldur~\citep{first2023baldur} generates and repairs whole Isabelle proofs with LLMs;
\citet{bursuc2025benchmark} introduce a large benchmark spanning Dafny, Verus, and Lean.
All these approaches require verification-aware languages or complete formal proofs.
\sys instead targets C with push-button bounded model checking, trading proof completeness
for practical accessibility.

\paragraph{Specification generation.}
Daikon~\citep{ernst2007daikon} discovers likely invariants from concrete executions,
but these are heuristically inferred and may not hold on unobserved inputs.
Lemur~\citep{wu2024lemur} integrates LLMs with automated verifiers to discharge invariants,
but targets full verification of existing programs rather than co-synthesis.
Recent work generates formal contracts~\citep{richter2025beyond} or Hoare-style
annotations~\citep{bouras2025hoareprompt} alongside code without formally verifying them.
\sys conducts best-effort formal verification for each generated specification.

\section{Conclusion}\label{sec:conclusion}

We presented \sys, an end-to-end pipeline that augments LLM-generated C programs with
formally verified annotations to aid developer comprehension.
Starting from a natural-language description, \sys co-generates a program and candidate
assertions, verifies them compositionally using a call-graph-aware procedure backed by a
portfolio of bounded model checkers, and surfaces only the verified facts to the developer
as natural-language comments.
On 17 diverse \texttocode benchmarks, \sys was able to successfully generate and verify a significant number of assertions. A user study with crowdworkers suggests that programs annotated with verified facts can improve answer accuracy and response speed on the selected code-comprehension tasks compared to unannotated programs.

\paragraph{Limitations and future work.}
The scope of \sys's guarantees is the bound-reduced program $\prog^+_k$, rather than directly the original program $\prog_0$ delivered to the user. Bound reduction may remove executions of $\prog_0$, so a property can hold on $\prog^+_k$ while failing on the original program unless the reduction preserves the relevant behaviors. This trade-off keeps verification feasible for complex programs, while making the annotations bounded evidence rather than end-to-end guarantees about $\prog_0$. In the future, it might be interesting to consider invoking AI-driven proof assistants to fully verify properties in a best-effort manner. Partial specifications may also mislead users into believing that the program as a whole is correct. Mitigating this risk may require interfaces that communicate the scope of each guarantee clearly, as well as training future software engineers to be more conscientious about, and familiar with, partial formal specifications. Verified annotations should complement, rather than replace, testing and human review.
While the scalability of bounded model checkers has improved significantly in the past two decades, \sys can still be bottlenecked by the underlying verifiers and can only work with programs supported by the verifiers. Extending the approach to richer languages such as Python or Rust, and to heap properties and quantified invariants, are natural next steps. Finally, the quality of verified annotations depends on the LLM generating meaningful assertions. In the future, it would be interesting to consider alternative property elicitation strategies or design metrics to evaluate the quality of the generated assertions. Due to budget constraints, the human study was conducted with crowdworkers and at a relatively small scale. While the patterns we observed are clear and consistent, expanding the user study with professional developers in an offline setting on a larger set of questions would further establish the effectiveness of \sys.

\section*{Acknowledgments}

We thank the anonymous NeurIPS reviewers, Steve Liang, and Manos Kapritsos for their engaging discussions. This work was supported by the VMware University Research Fund and performed in part using high-performance computing equipment obtained under National Science Foundation Grant No. 2117377. The authors declare no competing interests.
\newpage
%
%
%
\bibliographystyle{ACM-Reference-Format}
\bibliography{references}

\ifdefined\cameraready
  \input{camera_ready_disclosures.tex}
\fi

\ifdefined\cameraready
\else
  \newpage
  \appendix
  \section{Benchmark Descriptions}\label{app:benchmarks}

The following benchmarks are used in the evaluation (\S\ref{sec:eval-sat}). Each entry gives the full natural-language task description used as input to \sys.

\paragraph{bubblesort.}
Write a program that implements the Bubble Sort algorithm to sort an array of integers in ascending order.

\textbf{Program Specifications:}
\begin{itemize}
  \item \textbf{Input Handling:}
\begin{itemize}
  \item The program should first read an integer \texttt{n} (the number of elements), where $1 \le n \le 100$.
  \item The program should then read \texttt{n} integers from standard input to fill the array.
\end{itemize}
  \item \textbf{Algorithm Logic:}
\begin{itemize}
  \item Implement the standard Bubble Sort: repeatedly step through the list, compare adjacent elements, and swap them if they are in the wrong order (left > right).
\end{itemize}
  \item \textbf{Output:}
\begin{itemize}
  \item Once sorting is complete, print a final message "Sorted Array:" followed by the sorted elements.
  \item Print the total number of swaps performed during the entire execution.
\end{itemize}
\end{itemize}

\textbf{Example Input:}
5
5 1 4 2 8

\textbf{Example Output:}
Pass 1: 1 4 2 5 8
Pass 2: 1 2 4 5 8
Pass 3: 1 2 4 5 8
Sorted Array: 1 2 4 5 8
Total Swaps: 4

\paragraph{connected\_component.}
Write a program that reads an adjacency matrix of an undirected graph and identifies all its connected components.

The program should read from stdin:
\begin{itemize}
  \item An integer $V$, representing the number of vertices ($1 \le V \le 100$).
  \item $V \times V$ integers (0 or 1) representing the adjacency matrix.
\end{itemize}

Constraints and Logic:
\begin{enumerate}
  \item Input Validation:
\begin{itemize}
  \item The adjacency matrix must be symmetric. If $A[i][j] \neq A[j][i]$, print "Invalid adjacency matrix" and terminate.
\end{itemize}
  \item Output:
\begin{itemize}
  \item For each connected component, print the vertex indices (0-indexed) on a new line, sorted in ascending order.
  \item The components themselves should be printed in the order their first vertex was discovered.
\end{itemize}
\end{enumerate}

\noindent\textbf{Example Input:}
\begin{lstlisting}
4
0 1 0 0
1 0 0 0
0 0 0 1
0 0 1 0
\end{lstlisting}
\noindent\textbf{Example Output:}
\begin{lstlisting}
Component 1: 0 1
Component 2: 2 3
\end{lstlisting}

\paragraph{shortest\_path.}
Write a C program that reads an adjacency matrix for an unweighted graph and finds the shortest path between two specified vertices using Breadth-First Search (BFS).

The program should read from stdin:
\begin{itemize}
  \item An integer $V$, representing the total number of vertices ($2 \le V \le 100$).
  \item A $V \times V$ matrix of integers (0 or 1) representing the adjacency matrix, where 1 indicates an edge and 0 indicates no edge.
  \item Two integers, \texttt{start} and \texttt{end}, representing the source and destination vertex indices (0-indexed).
\end{itemize}

Constraints and Logic:
\begin{enumerate}
  \item Input Validation:
\begin{itemize}
  \item If \texttt{start} or \texttt{end} are outside the range $[0, V-1]$, print "Error: Invalid vertex index" and terminate.
  \item If $V$ is less than 2 or greater than 100, print "Error: Invalid graph size" and terminate.
\end{itemize}
  \item Shortest Path Logic:
\begin{itemize}
  \item Use a queue-based BFS to explore the graph starting from the \texttt{start} vertex.
  \item Maintain a \texttt{visited} array to avoid cycles and a \texttt{parent} array (initialized to -1) to track the path taken to reach each vertex.
  \item Since the graph is unweighted, the first time the \texttt{end} vertex is reached, the path found is guaranteed to be the shortest.
\end{itemize}
  \item Output:
\begin{itemize}
  \item If a path exists, print "Path: " followed by the sequence of vertices separated by arrows (e.g., 0 -> 2 -> 3).
  \item If no path exists between the two vertices, print "No path found".
\end{itemize}
\end{enumerate}

\noindent\textbf{Example Input:}
\begin{lstlisting}
4
0 1 1 0
1 0 0 1
1 0 0 1
0 1 1 0
0 3
\end{lstlisting}
\noindent\textbf{Example Output:}
\begin{lstlisting}
Path: 0 -> 1 -> 3
\end{lstlisting}

\paragraph{knapsack.}
Write a program that implements the Greedy algorithm for the Fractional Knapsack problem. The program should determine the maximum value that can be carried in a knapsack of a fixed capacity $W$ by taking items (or fractions of items) with the highest value-to-weight ratio first.

\textbf{Program Specifications:}
\begin{itemize}
  \item \textbf{Input Format:}
\begin{itemize}
  \item The first line of input contains an integer $n$ (number of items) and a double $W$ (maximum capacity).
  \item The following $n$ lines each contain two doubles: $v$ (value) and $w$ (weight) for that item.
\end{itemize}
  \item \textbf{Logic and Constraints:}
\begin{itemize}
  \item Calculate the density (value/weight) for each item.
  \item Sort the items in descending order based on their density.
  \item If an item fits completely, add its full value and weight to the knapsack.
  \item If an item does not fit completely, take the remaining fractional portion to fill the knapsack to exactly $W$ and stop.
  \item Handle cases where weight might be zero to avoid division errors.
\end{itemize}
  \item \textbf{Output:}
\begin{itemize}
  \item Print the total maximum value formatted to two decimal places.
  \item Print the breakdown of each item used: its original index, the amount of weight taken from that item, and its value contribution.
\end{itemize}
\end{itemize}

\textbf{Example Input:}
3 50
60 10
100 20
120 30

\textbf{Example Output:}
Total Value: 240.00
Item 1: Weight 10.00, Value 60.00
Item 2: Weight 20.00, Value 100.00
Item 3: Weight 20.00, Value 80.00

\paragraph{sudoku.}
Write a program that implements a playable 4×4 Sudoku game. The board is 4×4, using digits 1–4. The program starts with a partially filled puzzle (hardcoded).

On each turn:
\begin{itemize}
  \item Display the board.
  \item Prompt the user to enter row, column, and value (each in 1–4).
  \item A move is invalid if:
\begin{enumerate}
  \item The position is outside the board
  \item The value is not in 1–4
  \item The cell is pre-filled
  \item The value already appears in the same row, column, or 2×2 subgrid
\end{enumerate}
  \item Invalid moves print an error message and do not change the board.
  \item The game ends when the board is full and all Sudoku rules are satisfied, at which point the program prints a success message and terminates.
\end{itemize}

\paragraph{tictactoe.}
Write a C program that implements a playable two-player Tic-Tac-Toe game for the command line. The game is played on a 3×3 grid where Player 1 is 'X' and Player 2 is 'O'.

Requirements:
\begin{itemize}
  \item Initialize an empty 3×3 board and display it to the users.
  \item Alternatingly prompt Player 1 and Player 2 to enter their move as a pair of integers representing the row and column (1–3).
  \item Validate each move based on the following rules:
\begin{enumerate}
  \item The row and column must be within the range [1, 3].
  \item The chosen cell must not already be occupied by an 'X' or 'O'.
\end{enumerate}
  \item If a move is invalid, the program must print a descriptive error message and prompt the same player to try again without switching turns.
  \item After every valid move, the program must check for a win or a draw:
\begin{enumerate}
  \item A player wins if they have three of their marks in a horizontal, vertical, or diagonal row.
  \item A draw occurs if all 9 cells are filled and no player has won.
\end{enumerate}
  \item The game ends immediately when a player wins or a draw is reached, printing the final board state and a message announcing the winner (or a draw) before terminating.
  \item The program should on the high level be a loop that iterates for at most \texttt{MAX\_TURN} turns (e.g., 1000)
\end{itemize}

\paragraph{connect4.}
Write a program that implements a two-player Connect Four game on a 6×7 grid (6 rows, 7 columns).
\begin{itemize}
  \item Player 1 uses 'X', Player 2 uses 'O'. Player 1 goes first.
  \item The board should be initialized as empty (using '.' or spaces).
  \item The program should on the high level be a loop that iterates for at most \texttt{MAX\_TURN} turns (e.g., 1000)
\end{itemize}

\textbf{On each turn:}
\begin{itemize}
  \item Display the current state of the board.
  \item Prompt the current player to choose a column (1–7).
  \item A move is invalid if:
\begin{enumerate}
  \item The chosen column is outside the range 1–7.
  \item The chosen column is already full (contains 6 pieces).
\end{enumerate}
  \item If a move is invalid, print an error message and prompt the same player again.
  \item Gravity Logic: The piece must "fall" to the lowest available row in the chosen column.
\end{itemize}

\textbf{Winning and Ending:}
\begin{itemize}
  \item After each valid move, check if the current player has formed a line of four consecutive pieces:
\begin{itemize}
  \item Horizontally, vertically, or diagonally (both directions).
\end{itemize}
  \item If a player wins, display the final board, announce the winner ('X' or 'O'), and terminate.
  \item If the board is completely full and no player has four in a row, declare the game a draw and terminate.
\end{itemize}

\textbf{Example Board Display:}
. . . . . . .
. . . . . . .
. . . . . . .
. . . X . . .
. . O X . . .
. O X O . . .
1 2 3 4 5 6 7

\paragraph{numberle.}
Write a program that implements a "Numberle" game, where the user must guess a hidden 6-character mathematical equation.

\textbf{Game Rules:}
\begin{itemize}
  \item The target is a 6-character string representing a valid mathematical equation (e.g., \texttt{2+3=05} or \texttt{9-7==2}).
  \item For this program, the target equation should be hardcoded (e.g., \texttt{4*2=08}).
  \item The user has 6 attempts to guess the equation.
\end{itemize}

\textbf{On each turn:}
\begin{itemize}
  \item Prompt the user to enter a 6-character guess consisting of digits (0–9), operators (+, -, *, /), and an equals sign (=).
  \item \textbf{Validation:} A guess is invalid if:
\begin{enumerate}
  \item It is not exactly 6 characters long.
  \item It does not contain exactly one '=' sign.
  \item The mathematical statement is false (e.g., \texttt{2+2=05} is invalid because 2+2 does not equal 5).
\end{enumerate}
  \item If the guess is invalid, print an error message and do not count the turn.
\end{itemize}

\textbf{Feedback Logic:}
After a valid guess, compare it to the target and print a feedback string of 6 characters:
\begin{itemize}
  \item \texttt{G} (Green): The character is in the correct position.
  \item \texttt{Y} (Yellow): The character is in the equation but in the wrong position.
  \item \texttt{X} (Gray): The character is not in the equation at all.
\end{itemize}

\textbf{Ending the Game:}
\begin{itemize}
  \item If the user guesses the equation correctly, print a victory message and the number of tries taken.
  \item If the user fails to guess the equation within 6 tries, print the correct equation and a "Game Over" message.
\end{itemize}

\textbf{Example Interaction:}
Target: \texttt{4*2=08}
Guess 1: \texttt{1+5=06}
Feedback: \texttt{X X X G G X} (Only '=' and '0' are correct/present)

\paragraph{condorcet.}
Write a C program that determines the Condorcet winner of an election. A Condorcet winner is a candidate who wins a head-to-head density comparison against every other candidate. In other words, for a candidate *i* to be the Condorcet winner, for every other candidate *j*, more voters must rank *i* above *j* than rank *j* above *i*.

\textbf{Program Specifications:}
\begin{itemize}
  \item \textbf{Input Format:}
\begin{itemize}
  \item The program should first read an integer \texttt{n} (number of candidates) followed by \texttt{n} strings representing the names of the candidates.
  \item It should then read an integer \texttt{v} (number of voters).
  \item Finally, it should read \texttt{v} lines of input. Each line represents one ballot containing \texttt{n} integers, which are the indices (0 to n-1) of the candidates in order of preference (from most preferred to least preferred).
\end{itemize}
  \item \textbf{Logic:}
\begin{itemize}
  \item Use a 2D array or matrix to track pairwise "wins." For every pair of candidates (A, B), count how many voters preferred A over B.
  \item After processing all ballots, check if there is a single candidate who defeated every other opponent in their respective head-to-head matchups.
\end{itemize}
  \item \textbf{Constraints:}
\begin{itemize}
  \item Maximum candidates: 10.
  \item Maximum voters: 100.
  \item Candidate names: Max 31 characters, no spaces.
\end{itemize}
  \item \textbf{Output:}
\begin{itemize}
  \item If a Condorcet winner exists, print "Condorcet Winner: [Name]".
  \item If no candidate wins all head-to-head matchups (which can happen in a "Condorcet Paradox" or a tie), print "No Condorcet winner".
\end{itemize}
\end{itemize}

\textbf{Example Input:}
3
Alice Bob Charlie
3
0 1 2
0 2 1
1 2 0

\textbf{Example Output:}
Condorcet Winner: Alice

\paragraph{gaussian\_elimination.}
Write a C program that solves a system of linear equations using Gaussian elimination with partial pivoting.

The program should:
\begin{itemize}
  \item Read an integer \texttt{n} (the number of variables) from stdin.
  \item Read an augmented matrix of size \texttt{n × (n+1)} from stdin, where each row contains the coefficients of the equation and the constant term.
  \item Implement \textbf{partial pivoting} to ensure numerical stability. Before each elimination step, find the row with the largest absolute value in the current column and swap it with the current row.
  \item Perform forward elimination to reach an upper triangular matrix.
  \item Perform back substitution to solve for the variables $x_1, x_2, \dots, x_n$.
\end{itemize}

\textbf{Constraints and Rules:}
\begin{enumerate}
  \item Use \texttt{double} for all matrix calculations and variables.
  \item If the absolute value of a pivot element is less than $1e-9$ after attempting to pivot, the matrix is considered singular. In this case, print "No unique solution" and terminate the program.
  \item The output should display the solution for each variable $x_i$ (from $i = 1$ to $n$) on a separate line, formatted to 4 decimal places (e.g., \texttt{x1 = 1.2345}).
\end{enumerate}

\textbf{Example Input:}
3
2.0 1.0 -1.0 8.0
-3.0 -1.0 2.0 -11.0
-2.0 1.0 2.0 -3.0

\textbf{Example Output:}
x1 = 2.0000
x2 = 3.0000
x3 = -1.0000

\paragraph{simplex.}
Write a program that implements the Simplex algorithm to check the satisfiability of a set of linear equations. The program determines if there exists a non-negative vector $x \ge 0$ such that $Ax = b$.

The program should read from stdin:
\begin{itemize}
  \item Two integers $m$ and $n$, representing the number of constraints and variables respectively (max 10x10).
  \item $m$ lines representing the augmented matrix $[A | b]$.
\end{itemize}

Constraints and Logic:
\begin{enumerate}
  \item Initial Tableau: The program should assume the input is in standard form and identify an initial basic feasible solution.
  \item Pivoting Rule: Use Bland’s Rule (select the column/row with the lowest index) to prevent cycling during pivot operations.
  \item Operations:
\begin{itemize}
  \item Perform Gaussian elimination on the pivot element to update the tableau.
  \item If a row has all non-positive coefficients but a positive $b$ value, the system is UNSATISFIABLE.
  \item If the algorithm terminates with a feasible basis, the system is SATISFIABLE.
\end{itemize}
  \item Error Handling:
\begin{itemize}
  \item Print an error if $m$ or $n$ exceeds 10.
  \item Print an error if the input matrix is malformed.
\end{itemize}
\end{enumerate}

Output:
\begin{itemize}
  \item If a solution exists, print "SATISFIABLE" followed by the values of the variables $x_1, \dots, x_n$.
  \item If no solution exists, print "UNSATISFIABLE".
\end{itemize}

\noindent\textbf{Example Input:}
\begin{lstlisting}
2 4
1 0 2 1 10
0 1 3 2 15
\end{lstlisting}

\paragraph{log\_timestamp\_filter.}
rite a C program that reads a log file from \texttt{stdin} and prints all lines where the timestamp is strictly before the year 2024.

\noindent\textbf{Program Specifications:}
\begin{itemize}
  \item \textbf{Timestamp Format:} Each log line begins with a timestamp in the format \texttt{YYYY-MM-DD HH:MM:SS}.
  \item \textbf{Validation Rules:}
\begin{itemize}
  \item Each line must contain at most 248 characters (including the timestamp).
  \item The program must verify that the line starts with a valid timestamp format (at least 19 characters).
  \item Lines that exceed 248 characters or are too short to contain a timestamp should be ignored or handled gracefully without crashing.
\end{itemize}
  \item \textbf{Filtering Logic:}
\begin{itemize}
  \item Parse the year from the start of each line.
  \item If the year is numerically less than 2024 (e.g., 2023, 2022, etc.), print the entire line to \texttt{stdout}.
  \item If the year is 2024 or later, do not print the line.
\end{itemize}
  \item \textbf{Input Handling:} Read from standard input until \texttt{EOF}.
\end{itemize}

\noindent\textbf{Example Input:}
2023-12-31 23:59:59 System backup initiated
2024-01-01 00:00:05 Happy New Year log entry
2022-05-15 14:30:00 User "admin" logged in
2024-03-04 10:00:00 Routine maintenance check

\noindent\textbf{Example Output:}
2023-12-31 23:59:59 System backup initiated
2022-05-15 14:30:00 User "admin" logged in

\paragraph{sort\_time\_stamp.}
Write a C program that reads a list of ISO 8601 timestamps with various UTC offsets, sorts them in chronological order, and prints the sorted list.

The program should read from stdin:
\begin{itemize}
  \item An integer N, the number of timestamps (1 <= N <= 100).
  \item N lines, each containing a timestamp in the format: YYYY-MM-DDTHH:MM:SS[+/-]HH:MM
\end{itemize}
\noindent\textbf{Example: 2024-11-15T13:42:07+05:30}
\begin{lstlisting}

Constraints and Logic:
1. Normalization:
    - To compare timestamps accurately, the program must normalize each input to UTC (Offset +00:00).
    - If the offset is +05:30, subtract 5 hours and 30 minutes from the local time to get the UTC time.
    - Handle date rollovers (e.g., 2024-01-01T00:30:00+02:00 becomes 2023-12-31T22:30:00 UTC).
2. Sorting:
    - Sort the timestamps from earliest to latest based on their UTC value.
    - If two timestamps represent the exact same moment in UTC, maintain their relative input order (stable sort) or sort by the original string.
3. Memory and Parsing:
    - Each timestamp string is at most 32 characters.
    - Use a struct to store the original string and its calculated "Unix epoch" or normalized components for sorting.
4. Error Handling:
    - If a timestamp does not follow the format or contains invalid values (e.g., Month 13), print "Invalid Timestamp" and skip it.

Output:
- Print each original timestamp string on a new line, sorted chronologically.
\end{lstlisting}
\noindent\textbf{Example Input:}
\begin{lstlisting}
3
2024-11-15T10:00:00+00:00
2024-11-15T12:00:00+05:00
2024-11-15T06:00:00-02:00
\end{lstlisting}
\noindent\textbf{Example Output:}
\begin{lstlisting}
2024-11-15T12:00:00+05:00
2024-11-15T06:00:00-02:00
2024-11-15T10:00:00+00:00
\end{lstlisting}

\paragraph{shift\_letter.}
Write a program that reads a single line of text from stdin and prints a new string where every alphabetic character is shifted to the next letter in the alphabet.

On each execution:
\begin{itemize}
  \item The program must read the input until a newline character or EOF is encountered.
  \item For each character:
\begin{enumerate}
  \item If the character is a lowercase letter (a–z), replace it with the next letter in the alphabet. If the character is 'z', it should wrap around to 'a'.
  \item If the character is an uppercase letter (A–Z), replace it with the next letter in the alphabet. If the character is 'Z', it should wrap around to 'A'.
  \item Non-alphabetic characters (digits, punctuation, spaces) must remain unchanged.
\end{enumerate}
  \item The program should handle input strings up to 1024 characters in length.
  \item If the input exceeds 1024 characters, the program should print an error message "Input too long" and terminate with a non-zero exit code.
\end{itemize}

\noindent\textbf{Example:}
\begin{lstlisting}
Input:  "Abc-123 z!"
Output: "Bcd-123 a!"
\end{lstlisting}

\paragraph{base\_addition.}
Write a C program that performs addition on two positive integers represented as strings in a given base $n$, where $2 \le n \le 10$. The result should be returned in the same base.

\noindent\textbf{Input Specification:}
\begin{lstlisting}
- The program reads three space-separated strings from stdin: the base $n$, followed by two strings $S_1$ and $S_2$ representing the numbers.
- Each string $S_1$ and $S_2$ can be up to 256 characters long.

Logic and Constraints:
1. Validation:
    - If the base $n$ is not between 2 and 10, print "Error: Invalid Base" and terminate.
    - If any digit in $S_1$ or $S_2$ is greater than or equal to $n$, print "Error: Invalid Digit" and terminate.
2. Addition:
    - Perform the addition using a manual carry-over method (column-by-column) to accommodate the 256-character length.
    - The addition is performed in the specified base $n$ (e.g., in base 8, $7 + 1 = 10$).
3. Formatting:
    - The output should be a single string representing the sum.
\end{lstlisting}
\noindent\textbf{Example:}
\begin{lstlisting}
Input: 8 77 1
Output: 100

Input: 2 110 11
Output: 1001
\end{lstlisting}

\paragraph{cellular\_automata.}
Write a C program that performs an infinite simulation of Conway's Game of Life on a 2D toroidal grid. The program should continuously update the board and display the evolution in the terminal.

The program should read from stdin:
\begin{itemize}
  \item Two integers R and C, representing the number of rows and columns (max 50x50).
  \item R lines of C characters ('.' for dead, '\#' for alive) to set the initial state.
\end{itemize}

Constraints and Logic:
\begin{enumerate}
  \item Infinite Loop:
\begin{itemize}
  \item After reading the initial state, the program enters an infinite loop.
  \item On each iteration:
\end{itemize}
\end{enumerate}
a. Clear the terminal screen (using \verb|printf("\033[H\033[J");| or similar ANSI escape codes).
b. Display the current generation number and the grid state.
c. Calculate the next state based on toroidal rules (edges wrap around).
d. A cell lives if it has 2-3 neighbors, and a dead cell is born if it has exactly 3.
e. Wait for 100 milliseconds (using \texttt{usleep(100000)} from \texttt{<unistd.h>}) before the next iteration.
\begin{enumerate}
  \item Memory Management:
\begin{itemize}
  \item Use two 2D arrays to manage the current and next states to ensure simultaneous updates.
\end{itemize}
  \item Boundary Logic:
\begin{itemize}
  \item For a cell at (r, c), neighbors are calculated using modulo: \texttt{(r + dr + R) \% R} and \texttt{(c + dc + C) \% C}.
\end{itemize}
\end{enumerate}

\noindent\textbf{Example Input:}
\begin{lstlisting}
5 5
.....
.###.
.....
.....
.....
\end{lstlisting}
\noindent\textbf{Output Behavior:}
\begin{lstlisting}
- The terminal clears and displays the "Blinker" pattern alternating between horizontal and vertical states every 100ms indefinitely.
\end{lstlisting}

\paragraph{maze.}
Write a program that reads a maze from standard input and finds the shortest path from the start point 'S' to the end point 'E'.

\textbf{Program Specifications:}
\begin{itemize}
  \item \textbf{Input Format:}
\begin{itemize}
  \item The first line contains two integers, \texttt{R} and \texttt{C}, representing the number of rows and columns.
  \item The following \texttt{R} lines contain \texttt{C} characters representing the maze:
\begin{itemize}
  \item \texttt{\#} represents a wall (impassable).
  \item \texttt{.} represents an open path.
  \item \texttt{S} represents the starting position.
  \item \texttt{E} represents the exit.
\end{itemize}
\end{itemize}
  \item \textbf{Algorithm Requirements:}
\begin{itemize}
  \item Implement a Breadth-First Search (BFS) to guarantee the shortest path is found.
  \item You must maintain a queue of coordinates and a mechanism to track visited cells to avoid infinite loops.
\end{itemize}
  \item \textbf{Constraints and Logic:}
\begin{itemize}
  \item Movement is restricted to four directions: Up, Down, Left, and Right (no diagonals).
  \item The program must detect if no path exists between 'S' and 'E'.
  \item Max maze dimensions: 50x50.
\end{itemize}
  \item \textbf{Output:}
\begin{itemize}
  \item If a path is found, print the length of the shortest path (number of steps).
  \item Print the maze again, replacing the \texttt{.} characters on the optimal path with asterisks \texttt{*}.
  \item If no path exists, print "No path found".
\end{itemize}
\end{itemize}

\textbf{Example Input:}
5 5
S . . \# .
\noindent\textbf{\# . \# .}
. . . . .
. \# \# \# \#
. . . . E

\textbf{Example Output:}
Shortest Path: 8 steps
S * * \# .
\noindent\textbf{\# * \# .}
. . * * *
. \# \# \# *
. . . . E

\section{LLM Prompts}\label{app:prompts}

This appendix contains the four main prompts used by \sys.
Italicized placeholders (\textit{\{description\}}, \textit{\{program\}},
\textit{\{properties\}}) are filled at runtime with the corresponding artifact.

\subsection{Property Elicitation}\label{app:prompt-specs}

Used in Stage~1 (Section~\ref{sec:methodology}) to enumerate natural-language
safety and correctness properties from the task description.

\begin{lstlisting}
You are an expert C programmer. You need to generate a C program
based on the following natural language description:

  {description}

List safety/correctness properties of the program that performs
this task in succinct natural language. The properties should be
expressible as assertions in the code.

Your answer should just be:
1. ...
2. ...
\end{lstlisting}

\subsection{Verification-Friendly Program Synthesis}\label{app:prompt-synth}

Used in Stage~2 to generate the C program $\prog_0$.
The prompt enforces the structural constraints required for bounded verification.

\begin{lstlisting}
You are an expert C programmer and verification-aware developer.

Task: Generate a complete, working C program that implements the
following specification:

SPEC:
{description}

Constraints (must follow all):
1) Avoid using any third-party/external libraries.
2) Keep the program verification-friendly:
   - No recursion.
   - No dynamic allocation (malloc/free).
   - No floating point.
   - No pointer arithmetic beyond array indexing.
   - Avoid undefined behavior (signed overflow, out-of-bounds,
     shifting by >= width, uninitialized reads).
   - Use fixed maximum sizes for arrays/buffers; validate lengths.
   - For strings, use strnlen() instead of strlen().
3) Deterministic control flow for bounded verification:
   - Every loop must have clear static bounds (constants or
     validated input capped at a constant).
   - If a bound is configurable, declare it as a macro at the top
     (e.g., #define MAX_N 100).
4) Decomposable structure:
   - Provide small functions for each subtask (parsing, validation,
     core logic, output formatting).
5) I/O:
   - Read from stdin and write to stdout.
   - On invalid input, print an error and exit with nonzero code.
6) Output:
   - Return ONLY the full program as a single C file, wrapped in
     triple-backtick ```c formatting.
   - Include a brief comment at the top stating assumptions and bounds.
\end{lstlisting}

\subsection{Assertion Generation}\label{app:prompt-annot}

Used in Stage~3 to annotate $\prog_0$ with \texttt{assert()} statements
expressing the properties $\Pi$.

\begin{lstlisting}
You are an expert C programmer.

You are given:
  Description: {description}
  A C program that implements the task: {program}
  Safety and correctness properties: {properties}

Annotate the C program with assertions that express these properties
as function contracts.

Rules:
- Treat each C function (including static functions and main) as a
  method.
- Add preconditions as assert(...) at the very start of each function.
- Add postconditions as assert(...) immediately before each return.
  For void functions, place postconditions before the closing brace.
- Use only standard C assertions: assert(condition);
- Include <assert.h>.
- Do not add, remove, or modify existing code except for inserting
  assertions.
- Helper functions (returning bool, called within assert, containing
  no assert themselves) are allowed.
- Assertions may reference function parameters, return values (via
  existing variables), and globals in scope.

Applicability:
- Add an assertion only if it meaningfully applies to the function
  and can be soundly checked at the function boundary.
- Do not add trivial assertions (assert(true), assert(1), x == x).
- Only add assertions that enforce one of the listed properties.

Output:
- Output only the annotated C program in triple-backtick C formatting.
- Do not include any explanation outside the code block.
\end{lstlisting}

\subsection{Verified Fact Translation}\label{app:prompt-facts}

Used in Stage~5 to translate verified \texttt{assert()} statements into
natural-language \texttt{//FACT:} comments for the developer.

\begin{lstlisting}
You are an expert C programmer converting assertions in a C program
into easily understandable natural-language facts.

Goal: make assertions easier to understand for non-experts.
Facts should represent the technical content of the assertion but
be written in simple, accessible language.
Do not give MORE information than the assertion. Do not interpret
it---just translate it into simple natural language.

Format: "//FACT: At this point in the program, ..."

Replace each assertion statement with a //FACT: comment.
Keep all other code exactly the same.
Do not add code fences.

C Program:
{program}
\end{lstlisting}

\fi

\ifdefined\cameraready
  \newpage
  \section*{NeurIPS Paper Checklist}

\begin{enumerate}

\item {\bf Claims}
    \item[] Question: Do the main claims made in the abstract and introduction accurately reflect the paper's contributions and scope?
    \item[] Answer: \answerYes{}
    \item[] Justification: The abstract and introduction clearly state the three contributions: the co-generation pipeline, the compositional verification procedure, and the empirical evaluation (feasibility on 17 benchmarks + user study).

\item {\bf Limitations}
    \item[] Question: Does the paper discuss the limitations of the work performed by the authors?
    \item[] Answer: \answerYes{}
    \item[] Justification: Section~\ref{sec:conclusion} discusses the bounded nature of the guarantees, the risk of false confidence from partial specifications, verifier restrictions, dependence on LLM-generated assertions, and the use of crowdworkers.

\item {\bf Theory assumptions and proofs}
    \item[] Question: For each theoretical result, does the paper provide the full set of assumptions and a complete (and correct) proof?
    \item[] Answer: \answerYes{}
    \item[] Justification: Theorems state all assumptions in their formal statements, and their complete proofs appear in the main text.

    \item {\bf Experimental result reproducibility}
    \item[] Question: Does the paper fully disclose all the information needed to reproduce the main experimental results of the paper to the extent that it affects the main claims and/or conclusions of the paper (regardless of whether the code and data are provided or not)?
    \item[] Answer: \answerYes{}
    \item[] Justification: Section 4 describes the implementation details. Section 5.1 reports hyperparameter values and the experimental setups. Prompts and benchmark task descriptions appear in \paperappendicesref{app:benchmarks}{app:prompts}{A and B}.

\item {\bf Open access to data and code}
    \item[] Question: Does the paper provide open access to the data and code, with sufficient instructions to faithfully reproduce the main experimental results, as described in supplemental material?
    \item[] Answer: \answerNo{}
    \item[] Justification: The code is not publicly released. The paper provides algorithms, implementation details, prompts, and full benchmark specifications to support independent reimplementation, but these materials do not replace an executable artifact.

\item {\bf Experimental setting/details}
    \item[] Question: Does the paper specify all the training and test details (e.g., data splits, hyperparameters, how they were chosen, type of optimizer) necessary to understand the results?
    \item[] Answer: \answerYes{}
    \item[] Justification: No model training is involved. Section 5.1 specifies all relevant parameters in our experiments.

\item {\bf Experiment statistical significance}
    \item[] Question: Does the paper report error bars suitably and correctly defined or other appropriate information about the statistical significance of the experiments?
    \item[] Answer: \answerNo{}
    \item[] Justification: Verification results are deterministic per benchmark and do not admit error bars. We demonstrate the effectiveness of \sys over a wide range of tasks and consider a large number of users in the user studies.

\item {\bf Experiments compute resources}
    \item[] Question: For each experiment, does the paper provide sufficient information on the computer resources (type of compute workers, memory, time of execution) needed to reproduce the experiments?
    \item[] Answer: \answerYes{}
    \item[] Justification: Section 5.1 identifies the hardware platform (AMD Ryzen Threadripper PRO 5945WX workstation). Table 1 reports both LLM wall-clock time and verification time for every benchmark, giving a clear picture of end-to-end compute cost.
    
\item {\bf Code of ethics}
    \item[] Question: Does the research conducted in the paper conform, in every respect, with the NeurIPS Code of Ethics \url{https://neurips.cc/public/EthicsGuidelines}?
    \item[] Answer: \answerYes{}
    \item[] Justification: The user study involves only code-comprehension questions with no sensitive data collection; participants were compensated and participation was voluntary. No deception was used, and no harmful content was generated or released.

\item {\bf Broader impacts}
    \item[] Question: Does the paper discuss both potential positive societal impacts and negative societal impacts of the work performed?
    \item[] Answer: \answerYes{}
    \item[] Justification: Sections~\ref{sec:intro} and~\ref{sec:conclusion} discuss the benefit of trustworthy review aids and the principal harm of false confidence from incomplete positive facts, together with mitigations.
    
\item {\bf Safeguards}
    \item[] Question: Does the paper describe safeguards that have been put in place for responsible release of data or models that have a high risk for misuse (e.g., pre-trained language models, image generators, or scraped datasets)?
    \item[] Answer: \answerNA{}
    \item[] Justification: The paper poses no such risks.

\item {\bf Licenses for existing assets}
    \item[] Question: Are the creators or original owners of assets (e.g., code, data, models), used in the paper, properly credited and are the license and terms of use explicitly mentioned and properly respected?
    \item[] Answer: \answerYes{}
    \item[] Justification: CBMC, ESBMC, and Bitwuzla are cited in the paper and in the implementation section; all are open-source tools distributed under permissive licenses. The OpenAI API is used under its standard terms of service. No proprietary datasets are used.

\item {\bf New assets}
    \item[] Question: Are new assets introduced in the paper well documented and is the documentation provided alongside the assets?
    \item[] Answer: \answerNo{} 
    \item[] Justification: The paper introduces a prototype and experimental artifacts, but they are not currently released as documented public assets.

\item {\bf Crowdsourcing and research with human subjects}
    \item[] Question: For crowdsourcing experiments and research with human subjects, does the paper include the full text of instructions given to participants and screenshots, if applicable, as well as details about compensation (if any)? 
    \item[] Answer: \answerYes{}
    \item[] Justification: The paper describes the study design, eligibility criteria, and per-question compensation. The concrete survey is presented in the supplementary material.

\item {\bf Institutional review board (IRB) approvals or equivalent for research with human subjects}
    \item[] Question: Does the paper describe potential risks incurred by study participants, whether such risks were disclosed to the subjects, and whether Institutional Review Board (IRB) approvals (or an equivalent approval/review based on the requirements of your country or institution) were obtained?
    \item[] Answer: \answerNo{}
    \item[] Justification: The user study is through crowdsource workers and involves only voluntary code-comprehension questions with no sensitive data collection and no deception. Therefore, formal IRB review was not sought.

\item {\bf Declaration of LLM usage}
    \item[] Question: Does the paper describe the usage of LLMs if it is an important, original, or non-standard component of the core methods in this research? Note that if the LLM is used only for writing, editing, or formatting purposes and does \emph{not} impact the core methodology, scientific rigor, or originality of the research, declaration is not required.
    \item[] Answer: \answerYes{}
    \item[] Justification: LLMs are a central component of the methodology. Sections 3 and 4 describe each LLM call (property elicitation, synthesis, annotation, bound reduction, fact translation) and specify the configurations to reproduce the experiments.

\end{enumerate}

\fi

\end{document}